\documentclass[draft]{amsart}
\usepackage[foot]{amsaddr}

\usepackage{mathrsfs}
\usepackage{amssymb}

\usepackage[numbers]{natbib}

\def\hascitet{1}

\usepackage{mathtools}
\usepackage{etoolbox}
\usepackage{xparse}
\usepackage{xifthen}
\usepackage{scalerel}

\usepackage{upgreek}
\usepackage{bbm}

\usepackage{xfrac}
\usepackage{physics}
\usepackage{tensor}
\usepackage{derivative}

\usepackage{relsize}
\usepackage{scalefnt}
\usepackage[inline]{enumitem}
\usepackage{tasks}

\usepackage[super]{nth}

\DeclareMathAlphabet\mathbfcal{OMS}{cmsy}{b}{n}
\DeclareMathAlphabet\mathbfscr{OMS}{mdugm}{b}{n}

\newcommand\ii{\mathrm{i}}
\newcommand\mpi{\uppi}
\newcommand\ee{\mathrm{e}}

\let\oldva\va
\renewcommand\va[1] {\oldva{\boldsymbol #1}}

\let\oldvb\vb
\renewcommand\vb[1] {\oldvb{\boldsymbol #1}}

\let\oldvu\vu
\renewcommand\vu[1] {\oldvu{\boldsymbol #1}}

\newcommand\mat[1]  {\boldsymbol{#1}}

\newcommand\Natural{\mathbbm{N}}
\newcommand\Integer{\mathbbm{Z}}
\newcommand\Real{\mathbbm{R}}

\newcommand\Complex{\mathbbm{C}}

\DeclareDocumentCommand\rect{g}{\IfNoValueTF{#1}{\operatorname{rect}}{\fbraces{}{}{\operatorname{rect}}{#1}}}
\DeclareDocumentCommand\circ{g}{\IfNoValueTF{#1}{\operatorname{circ}}{\fbraces{}{}{\operatorname{circ}}{#1}}}

\DeclareMathOperator\dirac{\delta}
\DeclareDocumentCommand\kdelta{g}{\IfNoValueTF{#1}{\operatorname{\dirac}}{\operatorname{\dirac}_{#1}}}

\makeatletter
\newcommand*\transpose{{\mathpalette\@transpose{}}}
\newcommand*\@transpose[2]{\raisebox{\depth}{$\m@th#1\intercal$}}
\makeatother

\DeclarePairedDelimiterX{\definp}[2]{\langle}{\rangle}{#1\,\delimsize\vert\,#2}
\makeatletter
\def\inp{\@ifstar{\definp}{\definp*}}
\makeatother

\DeclareDocumentCommand\ensemble{sm}{
	\IfBooleanTF{#1}{\langle#2\rangle}{\left\langle#2\right\rangle}
}
\DeclareDocumentCommand\ensemblew{sm}{
	\IfBooleanTF{#1}{\langle#2\rangle_\omega}{\left\langle#2\right\rangle_\omega}
}
\DeclareDocumentCommand\tavg{sm}{
	\IfBooleanTF{#1}{\langle#2\rangle_\textrm{t}}{\left\langle#2\right\rangle_\textrm{t}}
}
\DeclareDocumentCommand\realization{o m}{\IfNoValueTF{#1}{\tensor*{#2}{}}{\tensor*[^{#1}]{#2}{}}}

\DeclareDocumentCommand\mathoperator{m m m}{
	\def\op{\mathchoice%
		{\displaystyle 		\scalebox{1.15}{$#1$}}%
		{\textstyle 		\scalebox{1.0} {$#1$}}%
		{\scriptstyle 		\scalebox{0.95}{$#1$}}%
		{\scriptscriptstyle	\scalebox{0.9} {$#1$}}%
	}
	\def\sub{\mathchoice%
		{\displaystyle 		#3}%
		{\textstyle 		#3}%
		{\scriptstyle 		\scalebox{.9}{$#3$}}%
		{\scriptscriptstyle	\scalebox{.8}{$#3$}}%
	}
	\tensor*{\op}{%
		_{\IfNoValueTF{#3}{}{\mkern-1mu\sub}}
		^{\IfNoValueTF{#2}{}{#2}}
	}
}
\DeclareDocumentCommand\fourieroperator{m m}{\mathoperator{\mathscr{F}}{#1}{#2}}
\DeclareDocumentCommand\frft{s o o g}{%
	\def\op{\fourieroperator{#2}{#3}}
	\def\body{#4 \vphantom{\fourieroperator{}{}}}
	\IfNoValueTF{#4}{\fourieroperator{#2}{#3}}{
		\IfBooleanTF{#1}
			{\op\lbrace\body\rbrace}
			{\fbraces{\lbrace}{\rbrace}{\op}{\body}}
	}%
}

\newcommand\prettyfrac[2]{\fontfamily{ppl}\selectfont\sfrac{#1}{#2}}
\newcommand\half[0]{\prettyfrac{1}{2}}

\makeatletter
\def\blfootnote{\gdef\@thefnmark{}\@footnotetext}
\makeatother

\usepackage{hyperref}
\usepackage[noabbrev,nameinlink,capitalize]{cleveref}
\crefname{subsection}{Subsection}{Subsections}
\Crefname{subsection}{Subsection}{Subsections}
\crefname{Table}{Table}{Tables}
\Crefname{Table}{Table}{Tables}
\creflabelformat{equation}{#2#1#3}

\newtheorem{theorem}{Theorem}[section]
\newtheorem{corollary}[theorem]{Corollary}
\newtheorem{lemma}[theorem]{Lemma}
\newtheorem{property}[theorem]{Property}

\crefname{task}{Property}{Properties}
\Crefname{task}{Property}{Properties}

\theoremstyle{plain}
\newtheorem*{remark}{Remark}

\DeclareMathOperator\HGop{\vb\Psi}
\DeclareMathOperator\HGdop{\widetilde{\vb\Psi}}
\DeclareDocumentCommand\hg{g}{%
	\tensor*{\Psi}{
		_{\IfNoValueTF{#1}{}{\mkern-1mu#1}}
	}
}
\DeclareDocumentCommand\HG{o g}{%
	\tensor*{\HGop}{
		_{\IfNoValueTF{#2}{}{\mkern-1mu#2}}
		^{\IfNoValueTF{#1}{}{#1}}
	}
}
\DeclareDocumentCommand\HGd{o g}{%
	\HGdop\!\tensor*{\vphantom{\HGop}}{
		_{\IfNoValueTF{#2}{}{\mkern-1mu#2}}
		^{\IfNoValueTF{#1}{}{#1}}
	}
}

\DeclareDocumentCommand\lctoperator{m m}{\mathoperator{\mathscr{L}}{#1}{#2}}
\DeclareDocumentCommand\lct{s o o g}{%
	\def\op{\lctoperator{#2}{#3}}
	\def\body{#4 \vphantom{\lctoperator{}{}}}
	\IfNoValueTF{#4}{\lctoperator{#2}{#3}}{
		\IfBooleanTF{#1}
			{\op\lbrace\body\rbrace}
			{\fbraces{\lbrace}{\rbrace}{\op}{\body}}
	}%
}
\DeclareDocumentCommand\fresneloperator{m m}{\mathoperator{\mathscr{O}}{#1}{#2}}
\DeclareDocumentCommand\fresnelt{s o o g}{%
	\def\op{\fresneloperator{#2}{#3}}
	\def\body{#4 \vphantom{\fresneloperator{}{}}}
	\IfNoValueTF{#4}{\fresneloperator{#2}{#3}}{
		\IfBooleanTF{#1}
			{\op\lbrace\body\rbrace}
			{\fbraces{\lbrace}{\rbrace}{\op}{\body}}
	}%
}
\DeclareDocumentCommand\laplaceoperator{m m}{\mathoperator{\mathcal{B}}{#1}{#2}}
\DeclareDocumentCommand\lt{s o o g}{%
	\def\op{\laplaceoperator{#2}{#3}}
	\def\body{#4 \vphantom{\laplaceoperator{}{}}}
	\IfNoValueTF{#4}{\laplaceoperator{#2}{#3}}{
		\IfBooleanTF{#1}
			{\op\lbrace\body\rbrace}
			{\fbraces{\lbrace}{\rbrace}{\op}{\body}}
	}%
}
\DeclareDocumentCommand\wvdoperator{m m}{\mathoperator{\mathscr{W}}{#1}{#2}}
\DeclareDocumentCommand\wvd{s o o g}{%
	\def\op{\wvdoperator{#2}{#3}}
	\def\body{#4 \vphantom{\wvdoperator{}{}}}
	\IfNoValueTF{#4}{\wvdoperator{#2}{#3}}{
		\IfBooleanTF{#1}
			{\op\lbrace\body\rbrace}
			{\fbraces{\lbrace}{\rbrace}{\op}{\body}}
	}%
}

\newcommand\pDv{{\boldsymbol\partial}}

\author{Shlomi Steinberg}
\author{\"{O}mer E\u{g}ecio\u{g}lu}
\author{Ling-Qi Yan}
\address{University of California, Santa Barbara, \\
         Department of Computer Science}
\email{p@shlomisteinberg.com, \{omer, lingqi\}@cs.ucsb.edu}

\title[The Anisotropic Multivariate Hermite-Gauss Functions]
      {On the Properties of the Anisotropic Multivariate Hermite-Gauss Functions}

\begin{document}



\begin{abstract}
	The Hermite-Gauss basis functions have been extensively employed in classical and quantum optics due to their convenient analytic properties.
	A class of multivariate Hermite-Gauss functions, the \emph{anisotropic Hermite-Gauss functions}, arise by endowing the standard univariate Hermite-Gauss functions with a positive definite quadratic form.
	These multivariate functions admit useful applications in optics, signal analysis and probability theory, however they have received little attention in literature.
	In this paper, we examine the properties of these functions, with an emphasis on applications in computational optics.
\end{abstract}

\maketitle


\section{Introduction} \label{section_introduction}

Different forms of the Hermite-Gauss functions have seen wide usage in physics and chemistry,
e.g., in the context of
detection of gravitational waves \cite{Ast_DiPace_Millo_Pichot_Turconi_Christensen_Chaibi_2021,Tao_Green_Fulda_2020},
quantum encoding \cite{Allgaier_Ansari_Donohue_Eigner_Quiring_Ricken_Brecht_Silberhorn_2020} and communication \cite{Perkins_Newell_Schabacker_Richardson_2013},
quantum entanglement with Hermite-Gauss beams \cite{Walborn_Pimentel_2012},
self-healing \cite{AguirreOlivas2015} and non-diffracting \cite{Chabou_Bencheikh_2020} (elegant) Hermite-Gauss beams,
detection beyond the diffraction limit \cite{Singh_Nagar_Roichman_Arie_2017},
Goos–Hänchen shift on reflection of a graphene monolayer \cite{Zhen_Deng_2020},
soft X-ray orbital angular momentum analysis \cite{Lee_Alexander_Kevan_Roy_McMorran_2019},
turbulence-resistant laser beams \cite{Cox_Maqondo_Kara_Milione_Cheng_Forbes_2019},
and for numeric integration \cite{HG_quadrature_6289843}.
This list is far from exhaustive.

The \emph{anisotropic Hermite-Gauss} (AHG) functions have been introduced by \citet{Amari_Kumon_1983}
(using the terminology ``tensorial Hermite-Gauss functions''), and were studied further later by \citet{jastor10.2307.25050684,Holmquist_1996,Ismail_Simeonov_2020}.
By using the quadratic form defined by a given positive definite matrix, these functions form a multivariate extension of the standard univariate Hermite-Gauss (HG) functions.
The positive definite matrix can be used for the representation of spatial deformations, geometric properties and energy tensors of structured optical beams, and potential other future applications.
In the context of optical coherence theory, it was shown that this anisotropy matrix has a clear physical meaning \cite{Steinberg_hg_2021}:
the spatial coherence of light.
This allows for the representation of a large family of coherence functions using a limited count of AHG modes, making the representation computationally-tractable.

\ifx\hascitet\undefined
Ismail and Simeonov \cite{Ismail_Simeonov_2020}
\else
\citet{Ismail_Simeonov_2020}
\fi
have derived certain properties of the AHG functions, including the generating functions, recurrence relations and linearization properties.
The purpose of this paper is to study the properties of these functions from a computational and optical perspective.
In addition to a number of useful identities, we derive closed-form expressions for the linear canonical transform (LCT) of an AHG function,
as well as two important transforms generalized by the LCT: the fractional Fourier transform
and Laplace transform.
In addition, we consider the Wigner-Vile distribution in Hermite-Gauss space.
These transforms are fundamental in Fourier optics, quantum mechanics and signal processing.
We also discuss the eigenfunctions of these transforms and show that the AHG functions are the eigenfunctions for specific cases of the LCT.
These results echo well-known results for the univariate HG functions which have not been previously investigated under the context of the multivariate AHG functions.


\section{Notation and Preliminaries} \label{section_preliminaries}

Let $\Natural=\qty{0,1,2,\ldots}$ represent the set of natural numbers and $\Integer,\Real,\Complex$ represent the set of integers, the real field and
the complex field, respectively.
A vector is denoted as $\va{r}=\qty[r_1,r_2,\ldots,r_n]^\transpose\in\Complex^n$ and the all-ones vector is denoted $\va{1}=\qty[1,1,\ldots,1]^\transpose\in\Complex^n$.
We use $\Real^{n\times m},\Complex^{n\times m}$ to denote the sets of all
real-valued and complex-valued $n\times m$
matrices, respectively.
Let $\mat{I}$ be the identity matrix, $\abs{\mat{A}}$ denote the determinant of a (square) matrix $\mat{A}$ and $\mat{A}^\transpose$ the transpose of $\mat{A}$.
Given $\mat{A}\in\Complex^{n\times m}$, the notation $\mat{A}=[\va{a}_j^\transpose]=[a_{jk}]$ defines $\va{a}_j^\transpose$, $a_{jk}$ to be the row vectors and elements
of $\mat{A}$, respectively.
A matrix $\mat{S}\in\Real^{n\times n}$ is said to be positive definite if it is symmetric and $\va{x}^\transpose\mat{S}\va{x}>0$ for all $0\neq\va{x}\in\Real^n$.
The notation $\mat{S}\succ 0$ indicates that $\mat{S}$ is positive definite.

A multi-index is defined as the $n$-tuple $\vb{\nu}=\qty(\nu_1,\nu_2,\ldots,\nu_n)\in\Natural^n$.
We use the standard multi-index factorial, double factorial, degree and power shorthand, viz.
\begin{alignat}{2}
	\vb{\nu}! &\triangleq \prod_j \nu_j!
	~,
	\qquad&&\qquad
	\vb{\nu}!! \triangleq \prod_j \nu_j!!
	~,
	\\
	\abs{\vb{\nu}} &\triangleq \sum_j \nu_j
	~,
	\qquad&&\qquad
	\va{r}^{\vb{\nu}} \triangleq \prod_j r_j^{\nu_j}
	~,
\end{alignat}
where the double factorial of a natural integer is $n!!=n\cdot(n-2)\cdot\ldots\cdot 1$ when $n$ is odd and $n!!=n\cdot(n-2)\cdot\ldots\cdot 2$ otherwise (the factorial and double factorial of 0 is 1).
The partial order $\preceq$ is defined on the set of multi-indices as follows: $\vb{\nu}\preceq\vb{\mu}$ iff $\forall_j \nu_j\leq\mu_j$.
The usual binomial coefficients are generalized to multi-indices as
\begin{align}
	\binom{\vb{\nu}}{\vb{\mu}} = \frac{\vb{\nu}!}{\vb{\mu}!\qty(\vb{\nu}-\vb{\mu})!}
	~,
\end{align}
the convention being that this binomial coefficient is non-zero iff $\vb{\nu}\preceq\vb{\mu}$.
For a multi-index $\vb{\nu}\in\Natural^n$ and a vector $\va{r}$, we define the partial derivative shorthand as
\begin{align}
	\pDv^{\vb{\nu}}_{\va{r}}
		&\triangleq
			\frac{\partial^{\abs{\vb{\nu}}}}{\prod_j \partial r_j^{\nu_j}}
	~.
\end{align}
Similarly, we define the multi-index matrix, $\mat{\Omega}\in\Natural^{n\times m}$, which consists of $n$ rows, each a multi-index, i.e $\mat{\Omega}=[\vb{\omega}_j]=[\omega_{jk}]$.
We define $\mat{\Omega}!=\prod_{j,k} \omega_{jk}!$ and,
given $\mat{A}=[a_{jk}]\in\Complex^{n\times m}$ set $\mat{A}^{\mat{\Omega}}=\prod_{j,k} a_{jk}^{\omega_{jk}}$.
We sometimes slightly abuse notation and write ${\va{1}}^{\intercal}\mat{\Omega}$ and $\mat{\Omega}\va{1}$ to denote the multi-indices that consist of the column sums and row sums of $\mat{\Omega}$, respectively.

Given a pair of $L^2$ functions $f,g$, the inner product (over $\Real^n$) of $f$ and $g$ is denoted by $\inp{f}{g} \triangleq \int_{\Real^n} \dd{\va{x}} f(\va{x})g^\star(\va{x})$, with $\star$ being complex conjugation.

\paragraph{\textbf{The Hermite-Gauss functions}}
The $k$-th order univariate, complex Hermite-Gauss function is defined as
\begin{align}
	\hg{k}\qty(z)
		&\triangleq
			\qty(\sqrt{\mpi} \, 2^k k!)^{-\half} \ee^{-\frac{z^2}{2}} H_k\qty(z)
		=
			\frac{
				\qty(-1)^k \ee^{\frac{z^2}{2}}
			}{\sqrt{\sqrt{\mpi} \, 2^k k!}}
			\dv[k]{z} \ee^{-z^2}
	~,
	\label{HG_univariate}
\end{align}
where $z\in\Complex$, $k\in\Natural$ and $H_k$ is the Hermite polynomial of order $k$.

Given a symmetric matrix $\mat{\Theta}\in\Complex^{n\times n}$ with a positive definite real part (i.e. $\Re\mat{\Theta}\succ 0$), we define the $n$-dimensional complex anisotropic Hermite-Gauss function of degree $\vb{\nu}\in\Natural^n$ of order $\abs{\vb{\nu}}$ associated to $\mat{\Theta}$ by
\begin{align}
	\HG[\mat{\Theta}]{\vb{\nu}}\qty(\va{r})
		&\triangleq
			\qty(-\frac{1}{\sqrt{2}})^{\abs{\vb{\nu}}}
			\frac{
				\ee^{\frac{1}{2}\va{r}^\transpose\mat{\Theta}^{-1}\va{r}}
			}{
				\sqrt{\vb{\nu}!}
				\qty(\mpi^n\abs{\mat{\Theta}})^{\frac{1}{4}}
			}
			\pDv^{\vb{\nu}}_{\va{r}}
				\ee^{-\va{r}^\transpose\mat{\Theta}^{-1}\va{r}}
	~.
	\label{hermite_gauss}
\end{align}
Similarly,
the \emph{dual} of the anisotropic Hermite-Gauss function is defined as
\begin{align}
	\HGd[\mat{\Theta}]{\vb{\nu}}\qty(\va{r})
		&\triangleq
			\qty(-\frac{1}{\sqrt{2}})^{\abs{\vb{\nu}}}
			\frac{
				\ee^{\frac{1}{2}\va{s}^\transpose\mat{\Theta}\va{s}}
			}{
				\sqrt{\vb{\nu}!}
				\qty(\mpi^n\abs{\mat{\Theta}})^{\frac{1}{4}}
			}
			\pDv^{\vb{\nu}}_{\va{s}}
				\ee^{-\va{s}^\transpose\mat{\Theta}\va{s}}
	~,
	\label{hermite_gauss_dual}
\end{align}
with $\va{s}=\mat{\Theta}^{-1}\va{r}$.
The generating functions of the AHG functions are
\begin{subequations}
\begin{align}
	\sum_{\vb{\nu}\in\Natural^n}
		\sqrt{\frac{2^{\abs{\vb{\nu}}}}{\vb{\nu}!}}
		\va{x}^{\vb{\nu}}
		\HG[\mat{\Theta}]{\vb{\nu}}\qty(\va{r})
	&=
		\frac{
			\ee^{-\frac{1}{2}\va{r}^\transpose\mat{\Theta}^{-1}\va{r} + \va{x}^\transpose\mat{\Theta}^{-1}\qty(2\va{r}-\va{x})}
		}{\qty(\mpi^n\abs{\mat{\Theta}})^{\frac{1}{4}}}
	\label{HG_generating_function}
	~,
	\\
	\sum_{\vb{\nu}\in\Natural^n}
		\sqrt{\frac{2^{\abs{\vb{\nu}}}}{\vb{\nu}!}}
		\va{x}^{\vb{\nu}}
		\HGd[\mat{\Theta}]{\vb{\nu}}\qty(\va{r})
	&=
		\frac{
			\ee^{-\frac{1}{2}\va{r}^\transpose\mat{\Theta}^{-1}\va{r} + \va{x}^\transpose\qty(2\va{r}-\mat{\Theta}\va{x})}
		}{\qty(\mpi^n\abs{\mat{\Theta}})^{\frac{1}{4}}}
	\label{HG_dual_generating_function}
	~,
\end{align}
\end{subequations}
for any $\va{x},\va{r}\in\Complex^n$ (see \cite{jastor10.2307.25050684,Ismail_Simeonov_2020}).


\section{Properties and Identities} \label{section_properties}

We begin with a few simple but useful properties of the AHG functions.
Most of the properties listed in \cref{properties_basic} are known
\cite{Ismail_Simeonov_2020,jastor10.2307.25050684} or easy to prove.
They are included here for completeness.
\begin{property}[Basic properties] \label{properties_basic}
	Let $\va{r}\in\Complex^n$, symmetric $\mat{\Theta}\in\Complex^{n\times n}$ such that $\Re\mat{\Theta}\succ 0$.
    Then
	\begin{tasks}[label-format=,label=\textrm{\thetheorem.\arabic*},label-width=3em,item-indent=5em](2)
		\task*	\label{basic_property_dual}
			$\HGd[\mat{\Theta}]{\vb{\nu}}\qty(\va{r}) =
				\abs{\mat{\Theta}}^{-\half}
				\HG[\mat{\Theta}^{-1}]{\vb{\nu}}\qty(\mat{\Theta}^{-1}\va{r})$.
		\task*	\label{basic_property_conjugate}
			$\HG[\mat{\Theta}]{\vb{\nu}}(\va{r})^\star = \HG[\mat{\Theta}^\star]{\vb{\nu}}(\va{r}^\star)$.
		\task*	\label{basic_property_analytic_b}
				$
					\HG[z^2\mat{\Theta}]{\vb{\nu}}\qty(\va{r}) =
						\abs{\mat{\Theta}}^{\frac{1}{4}}
						\abs{z^2\mat{\Theta}}^{-\frac{1}{4}}
						\qty(\frac{1}{z})^{\abs{\vb{\nu}}}
						\HG[\mat{\Theta}]{\vb{\nu}}\qty(\frac{1}{z}\va{r})
				$ for $0\neq z\in\Complex$.
		\task*	\label{basic_property_analytic_a}
			$\HG[\mat{\Theta}]{\vb{\nu}}\qty(-\va{r}) = (-1)^{\abs{\vb{\nu}}} \ii^n \HG[\mat{\Theta}]{\vb{\nu}}\qty(\va{r})$.
		\task*	\label{basic_property_real}
			if $\mat{\mat{\Theta}},\va{r}$ are real-valued then $\HG[\mat{\Theta}]{\vb{\nu}}\qty(\va{r})$ is real.
		\task*	\label{basic_property_decompose_into_univariate}
			if $\mat{\Theta}=\mat{I}$, the AHG function decomposes into a product of the univariate HG functions:
			$
				\HG[\mat{I}]{\vb{\nu}}(\va{r}) = \HGd[\mat{I}]{\vb{\nu}}(\va{r}) = \prod_k \hg{\nu_k}(r_k)
			$.
		\task*	\label{basic_property_even_odd}
			$\HG[\mat{I}]{\vb{\nu}}(\va{r})$ is even as a function of $r_j$ iff $\nu_j$ is even, otherwise it is odd.
	\end{tasks}
\end{property}
\begin{proof}
	\cref{basic_property_dual,basic_property_decompose_into_univariate} follow trivially from the definitions.
	\cref{basic_property_analytic_b,basic_property_conjugate} follow from the generating function (\cref{HG_generating_function}).
	\cref{basic_property_real} is a consequence of \cref{basic_property_conjugate}.
	\cref{basic_property_analytic_a} is a special case of \cref{basic_property_analytic_b}.
	\cref{basic_property_even_odd} is a consequence of \cref{basic_property_decompose_into_univariate} and the fact that the univariate HG function $\hg{k}$ is even iff $k$ is even and odd otherwise.
\end{proof}

\begin{property}[Derivatives]
	Let $\mat{\Theta}^{-1}=[\va{q}_j]$ be the rows of $\mat{\Theta}^{-1}$.
	Then the partial derivative, gradient, Hessian matrix and Laplacian of the AHG function are given by
	\begin{tasks}[label-format=,label=\textrm{\thetheorem.\arabic*},label-width=3em,item-indent=5em]
		\task \label{derivative_property_partial}
			$
				\pdv{}{r_j} \HG[\mat{\Theta}]{\vb{\nu}}\qty(\va{r})
					 + \va{q}_j^\transpose\va{r} \HG[\mat{\Theta}]{\vb{\nu}}\qty(\va{r})
					= 2 \va{q}_{j}^\transpose\va{\phi}_{\vb{\nu}}
			~,$
		\task \label{derivative_property_vec}
			$
				\pdv{}{\va{r}} \HG[\mat{\Theta}]{\vb{\nu}}\qty(\va{r})
					= \mat{\Theta}^{-1}
						\qty[
							2\va{\phi}_{\vb{\nu}}
							- \va{r}\HG[\mat{\Theta}]{\vb{\nu}}\qty(\va{r})
						]
			~,$
		\task \label{derivative_property_hessian}
			$
            \begin{aligned}[t]
				\pdv[order={2}]{}{\va{r}} \HG[\mat{\Theta}]{\vb{\nu}}\qty(\va{r})
					=&
						-\mat{\Theta}^{-1} \HG[\mat{\Theta}]{\vb{\nu}}\qty(\va{r})
                    \\
					&	+
						2\mat{\Theta}^{-1}
							\qty(
								\mat{\Phi}_{\vb{\nu}}
								- \va{r}\va{\phi}_{\vb{\nu}}^\transpose
								- \va{\phi}_{\vb{\nu}}\va{r}^\transpose
								+ \frac{1}{2}
									\va{r}\va{r}^\transpose
									\HG[\mat{\Theta}]{\vb{\nu}}\qty(\va{r})
							)\mat{\Theta}^{-1}
			~, 
            \end{aligned}
            $
		\task \label{derivative_property_laplacian}
			$
            \begin{aligned}[t]
				\laplacian{\HG[\mat{\Theta}]{\vb{\nu}}\qty(\va{r})}
					=&
						-
						\HG[\mat{\Theta}]{\vb{\nu}}\qty(\va{r})
						\tr\mat{\Theta}^{-1}
						+
						2\tr\qty(
							\mat{\Theta}^{-2}
							\mat{\Phi}_{\vb{\nu}}
						)
                    \\
					&	+
						\qty(\mat{\Theta}^{-1}\va{r})^\transpose
						\qty[
							\HG[\mat{\Theta}]{\vb{\nu}}\qty(\va{r})
							\mat{\Theta}^{-1}\va{r}
							-
							4 \mat{\Theta}^{-1}\va{\phi}_{\vb{\nu}}
						]
			~,
            \end{aligned}
            $
	\end{tasks}
	where 
	$\laplacian=\sum_j\pdv[order={2}]{}{r_j}$ is the Laplace operator (taken with respect to $\va{r}$),
$\pdv[order={2}]{}{\va{r}}$ is the Hessian, the matrix $\mat{\Phi}_{\vb{\nu}}$ is given by
\cref{derivate_eq_Phi} and with
	\begin{align}
		\va{\phi}_{\vb{\nu}}
			&=
				\frac{1}{\sqrt{2}}
				\begin{bmatrix}
					\sqrt{\nu_1} \HG[\mat{\Theta}]{\vb{\nu}-\vb{\varepsilon}_1}\qty(\va{r}), &
					\sqrt{\nu_2} \HG[\mat{\Theta}]{\vb{\nu}-\vb{\varepsilon}_2}\qty(\va{r}), &
					\ldots, &
					\sqrt{\nu_n} \HG[\mat{\Theta}]{\vb{\nu}-\vb{\varepsilon}_n}\qty(\va{r})
				\end{bmatrix}^\transpose
		~,
	\end{align}
	where $\vb{\varepsilon}_k\in\Natural^n$ is such that $(\varepsilon_k)_j=\kdelta{jk}$, i.e. the multi-index with $1$ at position $k$ and 0 elsewhere.
	\begin{remark}
        We adopt the convention that the AHG function vanishes identically if its degree contains negative elements.
	\end{remark}
	\begin{remark}
		\cref{derivative_property_vec,derivative_property_partial} were first derived by \citet{jastor10.2307.25050684}.
		A proof is provided below for completeness.
	\end{remark}
\end{property}
\begin{proof}
Differentiate the generating function (\cref{HG_generating_function}):
	\begin{align}
		\sum_{\vb{\nu}\in\Natural^n}
			\sqrt{\frac{2^{\abs{\vb{\nu}}}}{\vb{\nu}!}}
			\va{x}^{\vb{\nu}}
			\pdv{}{r_j} \HG[\mat{\Theta}]{\vb{\nu}}\qty(\va{r})
		&=
			\va{q}_{j}^\transpose\qty(2\va{x} - \va{r})
			\frac{
				\ee^{-\frac{1}{2}\va{r}^\transpose\mat{\Theta}^{-1}\va{r} + \va{x}^\transpose\mat{\Theta}^{-1}\qty(2\va{r}-\va{x})}
			}{\qty(\mpi^n\abs{\mat{\Theta}})^{\frac{1}{4}}}
		\\
		&=
		\va{q}_{j}^\transpose\qty(2\va{x} - \va{r})
		\sum_{\vb{\nu}\in\Natural^n}
			\sqrt{\frac{2^{\abs{\vb{\nu}}}}{\vb{\nu}!}}
			\va{x}^{\vb{\nu}}
			\HG[\mat{\Theta}]{\vb{\nu}}\qty(\va{r})
	\end{align}
	and equate the powers of $\va{x}$ on both sides, proving \cref{derivative_property_partial}.
	\cref{derivative_property_vec} follows immediately from \cref{derivative_property_partial}.

	Differentiate \cref{derivative_property_vec}:
	\begin{align}
		\pdv[order={2}]{}{\va{r}} \HG[\mat{\Theta}]{\vb{\nu}}\qty(\va{r})
			&=
				\mat{\Theta}^{-1}
					\pdv{}{\va{r}}\qty[
						2\va{\phi}_{\vb{\nu}}
						- \va{r}\HG[\mat{\Theta}]{\vb{\nu}}\qty(\va{r})
					]
		\\
			&=
				\mat{\Theta}^{-1}
					\qty[
						2\mat{\Phi}_{\vb{\nu}}
						- \HG[\mat{\Theta}]{\vb{\nu}}\qty(\va{r})
							\qty[\mat{I} - \va{r}\qty(\mat{\Theta}^{-1}\va{r})^\transpose]
						- 2\va{r}\qty(\mat{\Theta}^{-1}\va{\phi}_{\vb{\nu}})^\transpose
					],
	\end{align}
	where $\mat{\Phi}_{\vb{\nu}}=\pdv{}{\va{r}}\va{\phi}_{\vb{\nu}}$ is the matrix with the following elements:
	\begin{align}
		\qty[\mat{\Phi}_{\vb{\nu}}]_{jk}
			&\triangleq
				\begin{cases}
					\sqrt{\nu_j\qty(\nu_j-1)} \HG[\mat{\Theta}]{\vb{\nu}-2\vb{\varepsilon}_j}\qty(\va{r})			& \qif j=k	  \\
					\sqrt{\nu_j\nu_k} \HG[\mat{\Theta}]{\vb{\nu}-\vb{\varepsilon}_j-\vb{\varepsilon}_k}\qty(\va{r}) & \qotherwise
				\end{cases}
		\label{derivate_eq_Phi}
	\end{align}
	and simplify, yielding \cref{derivative_property_hessian}.

To complete the proof,
	note that  $\laplacian\equiv\tr\pdv[order={2}]{}{\va{r}}$ and recall
that the trace of an outer product is the inner product. This gives \cref{derivative_property_laplacian}.
\end{proof}

\begin{lemma}[Orthogonality and completeness] \label{lemma_orthogonality}
	Given a symmetric matrix $\mat{\mat{\Theta}}$ with a positive definite real part, the anisotropic Hermite-Gauss functions $\HG[\mat{\Theta}]{\vb{\nu}}$ form a complete orthonormal (with respect to their dual) basis of $\Real^n\to\Complex$ $L^2$-functions.
	In other words,
	\begin{enumerate}
		\item
			For all $\vb{\nu},\vb{\mu}\in\Natural^n$, $\inp{\HG[\mat{\Theta}]{\vb{\nu}}}{\HGd[\mat{\Theta}]{\vb{\mu}}} = \kdelta{\vb{\nu}\vb{\mu}}$,
			where $\kdelta$ denotes the Kronecker delta; and
		\item
			If an $L^2$-function $f$ is orthogonal to all $\HG[\mat{\Theta}]{\vb{\nu}}$, then $f$ vanishes a.e.
	\end{enumerate}
\end{lemma}
\begin{proof}
	See 
    \ifx\hascitet\undefined
    Ismail and Simeonov \cite{Ismail_Simeonov_2020}.
    \else
    \citet{Ismail_Simeonov_2020}.
    \fi
\end{proof}

Our main contributions in this section start with the next lemma, which allows for the
expansion of an AHG function as a finite series of AHG functions with different anisotropy.
\begin{lemma}[Anisotropy transformation] \label{lemma_HG_anisotropy_transform_identity}
	Given symmetric $\mat{\Theta}_1,\mat{\Theta}_2\in\Complex^{n\times n}$, with
$\Re\mat{\Theta}_1,\Re\mat{\Theta}_2\succ 0$, we have
	\begin{align}
		\HG[\mat{\Theta}_1]{\vb{\nu}}\qty(\va{r})
			&=
		\sqrt{\vb{\nu}! \abs{\mat{T}}}
		\sum_{\substack{
				\mat{\Omega}=\qty[\vb{\omega}_{j}^\transpose]\in\Natural^{n\times n}
				\text{, s.t. }
				\va{1}^\transpose\mat{\Omega}=\vb{\nu}
				\\
				\mathclap{
					\text{with }
					\vb{\mu} = \qty(\abs{\vb{\omega}_{1}},\abs{\vb{\omega}_{2}},\ldots,\abs{\vb{\omega}_{n}})
				}
			}}
			\frac{\mat{T}^{\mat{\Omega}}}
				 {\mat{\Omega}!}
			\sqrt{\vb{\mu}!}
			\HG[\mat{\Theta}_2]{\vb{\mu}}\qty(\mat{T}\va{r})
		\label{transformation_HG_identity}
	~,
	\end{align}
	where $\mat{T}=\mat{\Theta}_2^{\half}\mat{\Theta}_1^{-\half}$.
	The summation is over all $n\times n$ multi-index matrices $\mat{\Omega}$, with rows $\vb{\omega}_j$, such that the sum of the $k$-th column of $\mat{\Omega}$ is $\nu_k$.
	The multi-index $\vb{\mu}\in\Natural^n$ is defined to be the row sums of $\mat{\Omega}$.\\
\begin{remark}
	There are $\prod_j p(\nu_j)$ such matrices, where $p(m)$ is the partition function, which asymptotically grows as $\order{\exp(\sqrt{\abs{\vb{\nu}}})}$.
\end{remark}
\end{lemma}
\begin{proof}
	Start with the AHG generating function, \cref{HG_generating_function}, and perform the variable changes $\va{y}=\mat{T}\va{x}$ and $\va{s}^\prime=\mat{T}\va{r}$, viz.
	\begin{align}
		\sum_{\vb{\nu}\in\Natural^n}
				\sqrt{\frac{2^{\abs{\vb{\nu}}}}{\vb{\nu}!}}
				\va{x}^{\vb{\nu}}
				\HG[\mat{\Theta}_1]{\vb{\nu}}\qty(\va{r})
			&=
				\frac{
					\ee^{-\frac{1}{2}\va{s}^\transpose\mat{\Theta}_2^{-1}\va{s} + \va{y}^\transpose\mat{\Theta}_2^{-1}\qty(2\va{s}-\va{y})}
				}{\qty(\mpi^n\abs{\mat{\Theta}_1})^{\frac{1}{4}}}
            \\
			&=
                \abs{\mat{T}}^{\half}
                \sum_{\vb{\nu}\in\Natural^n}
                        \sqrt{\frac{2^{\abs{\vb{\nu}}}}{\vb{\nu}!}}
                        \va{y}^{\vb{\nu}}
                        \HG[\mat{\Theta}_2]{\vb{\nu}}\qty(\va{s})
	~.
	\end{align}
	Then, by the multinomial theorem:
	\begin{align}
		y_k^{\nu_k} &=
			\sum_{\substack{\;\;\vb{\omega}\in\Natural^n\\\abs{\vb{\omega}}=\nu_k}}
				\frac{\nu_k!}{\vb{\omega}!}
				\va{t}_k^{\mkern2mu \vb{\omega}}
				\va{x}^{\mkern1mu \vb{\omega}}
	~,
	\end{align}
	where the summation is over all the integer partitions of $\nu_k$ and we denote $\mat{T}=\qty[\va{t}_{j}^\transpose]$, i.e. $\va{t}_j$ are the rows of $\mat{T}$.
	The two equations above yield
	\begin{align}
		\sum_{\vb{\nu}\in\Natural^n}
				\sqrt{\frac{2^{\abs{\vb{\nu}}}}{\vb{\nu}!}}
				\va{x}^{\vb{\nu}}
				\HG[\mat{\Theta}_1]{\vb{\nu}}\qty(\va{r})
			&=
		\abs{\mat{T}}^{\half}
		\sum_{\substack{
				\mat{\Omega}=\qty[\vb{\omega}_{j}^\transpose]\in\Natural^{n\times n},
				\\
				\mathclap{
					\text{with }
					\vb{\mu} = \qty(\abs{\vb{\omega}_{1}},\abs{\vb{\omega}_{2}},\ldots,\abs{\vb{\omega}_{n}})
				}
			}}
			\sqrt{2^{\abs{\vb{\mu}}}\vb{\mu}!}
			\HG[\mat{\Theta}_2]{\vb{\mu}}\qty(\va{s})
			\prod_k
				\frac{
					\va{t}_k^{\mkern2mu \vb{\omega}_{k}}
					\va{x}^{\mkern1mu \vb{\omega}_{k}}
				}{\vb{\omega}_{k}!}
		~.
	\end{align}
	Equating the powers of $\va{x}$ on both sides above gives \cref{transformation_HG_identity}.
\end{proof}

Immediate consequences of the above lemma are the next few corollaries.
The first corollary facilitates the dimensional decomposition of an arbitrary AHG function into (finite) univariate HG functions.
This has useful computational applications.
\begin{corollary}[Dimensional decomposition] \label{corollary_dimensional_decomposition}
	With $\va{s}=\mat{\Theta}^{-\half}\va{r}$,
	\begin{align}
		\HG[\mat{\Theta}]{\vb{\nu}}\qty(\va{r})
			&=
		\sqrt{\vb{\nu}!}
		\abs{\mat{\Theta}}^{-\frac{1}{4}}
		\sum_{\substack{
				\mat{\Omega}=\qty[\vb{\omega}_{j}^\transpose]\in\Natural^{n\times n}
				\text{, s.t. }
				\va{1}^\transpose\mat{\Omega}=\vb{\nu}
				\\
				\mathclap{
					\text{with }
					\vb{\mu} = \qty(\abs{\vb{\omega}_{1}},\abs{\vb{\omega}_{2}},\ldots,\abs{\vb{\omega}_{n}})
				}
			}}
			\frac{\qty(\mat{\Theta}^{-\half})^{\mat{\Omega}}}
				 {\mat{\Omega}!}
			\sqrt{\vb{\mu}!}
			\prod_k
				\hg{\mu_k}\qty(s_k)
		.
	\end{align}
\end{corollary}

It is often important to evaluate the AHG functions at 0, e.g., for computation of the peak energy of optical beams or the determination of
the total energy carried by a wave ensemble \cite{Steinberg_hg_2021}.
The next corollary provides an explicit expression for the values at $0$ and may admit interesting combinatorics.
\begin{corollary}[The AHG function at 0] \label{corollary_hg_at_zero}
	Applying \cref{corollary_dimensional_decomposition} and recalling the values of the Hermite polynomials at 0, viz. $H_k(0)=(-2)^{\frac{k}{2}}(k-1)!!$ when $k$ is even and $\hg{k}(0)=0$ when $k$ is odd, results in
	\begin{align}
		\HG[\mat{\Theta}]{\vb{\nu}}\qty(0)
			&=
		\frac{\sqrt{\vb{\nu}!}}
			 {\qty(\mpi^n\abs{\mat{\Theta}})^{\frac{1}{4}}}
		\sum_{\substack{
				\mat{\Omega}=\qty[\vb{\omega}_{j}^\transpose]\in\Natural^{n\times n}
				\text{, s.t. }
				\va{1}^\transpose\mat{\Omega}=\vb{\nu}
				\\
				\mathclap{
					\text{and }
					\vb{\mu} = \qty(\abs{\vb{\omega}_{1}},\abs{\vb{\omega}_{2}},\ldots,\abs{\vb{\omega}_{n}})\in(2\Natural)^n
				}
			}}
			\frac{\qty(\mat{\Theta}^{-\half})^{\mat{\Omega}}}
				 {\mat{\Omega}!}
			\ii^{\abs{\vb{\mu}}}
			\qty(\vb{\mu}-\vb{1})!!
		,
	\end{align}
	with $\vb{1}=\qty(1,1,\ldots,1)\in\Natural^n$.
	\begin{remark}
		Note that the summation is now also constrained to multi-index matrices with even row sums.
		The double factorial of $-1$ is defined to be 1.
	\end{remark}
\end{corollary}

\begin{lemma}[Offseted argument] \label{lemma_offseted}
	For an arbitrary $\va{s}\in\Complex^n$:
	\begin{align}
		\HG[\mat{\Theta}]{\vb{\nu}}\qty(\va{r}+\va{s})
			&=
				2^{-\frac{\abs{\vb{\nu}}}{2}}
				\qty(\mpi^n \abs{\mat{\Theta}})^{\frac{1}{4}}
				\ee^{\frac{1}{2} \qty(\va{r}-\va{s})^\transpose \mat{\Theta}^{-1}\qty(\va{r}-\va{s})}
            \nonumber
            \\
            &\qquad\qquad\quad\times
				\sum_{\substack{
					\vb{\mu}\in\Natural^n \\
					\mathclap{
						\text{s.t. }
						\vb{\mu}\preceq\vb{\nu}
					}
				}}
					{\binom{\vb{\nu}}{\vb{\mu}}}^{\half}
					\HG[\mat{\Theta}]{\vb{\nu}-\vb{\mu}}\qty(\sqrt{2}\va{r})
					\HG[\mat{\Theta}]{\vb{\mu}}\qty(\sqrt{2}\va{s})
		~.
	\end{align}
\end{lemma}
\begin{proof}
	Via the generating function:
	\begin{align}
		\!\!\sum_{\vb{\nu}\in\Natural^n}
			\sqrt{\frac{2^{\abs{\vb{\nu}}}}{\vb{\nu}!}}
			\va{x}^{\vb{\nu}}
			&\HG[\mat{\Theta}]{\vb{\nu}}\qty(\va{r}+\va{s})
		=
			\frac{
				\ee^{-\frac{1}{2}\qty(\va{r}+\va{s})^\transpose\mat{\Theta}^{-1}\qty(\va{r}+\va{s}) + \va{x}^\transpose\mat{\Theta}^{-1}\qty(2\qty(\va{r}+\va{s})-\va{x})}
			}{\qty(\mpi^n\abs{\mat{\Theta}})^{\frac{1}{4}}}
		\\
		&=
			\frac{
					\ee^{-\frac{1}{2}\qty(\va{r}-\va{s})^\transpose\mat{\Theta}^{-1}\qty(\va{r}-\va{s})}
				 }
				 {\qty(\mpi^n\abs{\mat{\Theta}})^{\frac{1}{4}}}
			\ee^{-\frac{1}{2}\va{r}^\transpose\qty(\frac{1}{2}\mat{\Theta})^{-1}\va{r} +
				  \frac{\va{x}^\transpose}{2}(\frac{1}{2}\mat{\Theta})^{-1}\qty(2\va{r}-\frac{\va{x}}{2})}
        \nonumber \\
        &\qquad\qquad
        \times
			\ee^{-\frac{1}{2}\va{s}^\transpose\qty(\frac{1}{2}\mat{\Theta})^{-1}\va{s} +
				  \frac{\va{x}^\transpose}{2}(\frac{1}{2}\mat{\Theta})^{-1}\qty(2\va{s}-\frac{\va{x}}{2})}
		\\
		&=
			\frac{
					\qty(\mpi^n\abs{\mat{\Theta}})^{\frac{1}{4}}
				 }
				 {2^{\frac{n}{2}}}
			\ee^{-\frac{1}{2}\qty(\va{r}-\va{s})^\transpose\mat{\Theta}^{-1}\qty(\va{r}-\va{s})}
        \nonumber
        \\
        &\qquad \times
				\sum_{\vb{\nu},\vb{\mu}\in\Natural^n  \vphantom{\vb{\mu}}}
					\sqrt{\frac{2^{\abs{\vb{\nu}}}}{\vb{\nu}!}
						  \frac{2^{\abs{\vb{\mu}}}}{\vb{\mu}!}}
					\qty(\frac{1}{2}\va{x})^{\vb{\nu}+\vb{\mu}}
					\HG[\frac{1}{2}\mat{\Theta}]{\vb{\nu}}\qty(\va{r})
					\HG[\frac{1}{2}\mat{\Theta}]{\vb{\mu}}\qty(\va{s})
		~.
	\end{align}
	Equating the powers of $\va{x}$ and applying \cref{basic_property_analytic_b} yields the desired result.
\end{proof}

\begin{lemma}[Product of AHG functions] \label{lemma_product_of_hg}
	\begin{align}
		\HG[\mat{\Theta}]{\vb{\nu}}\qty(\va{r})
		\HG[\mat{\Theta}]{\vb{\mu}}\qty(\va{r})
			=&
				\sqrt{\frac{\vb{\nu}!\vb{\mu}!}{2^{\abs{\vb{\nu}}+\abs{\vb{\mu}}}}}
				\frac{\ee^{-\frac{1}{2}\va{r}^\transpose\mat{\Theta}^{-1}\va{r}}}
					 {\qty(\mpi^n\abs{\mat{\Theta}})^{\frac{1}{4}}}
            \nonumber 
            \\
            &\quad\times
				\sum_{\substack{
					\qquad \mat{\Omega}\in\Natural^{n\times n} ,\qquad \\
					\mathclap{
						\text{s.t.}\; \vb{\beta} = \vb{\nu} - \mat{\Omega}\va{1} \in \Natural^n,
					}\\
					\mathclap{
						\;\  \vb{\gamma} = \vb{\mu} - \va{1}^\transpose\mat{\Omega} \in \Natural^n
					}
				}}
				\frac{
						\qty(2\mat{\Theta}^{-1})^{\mat{\Omega}}
						\sqrt{2^{\abs{\vb{\beta}}+\abs{\vb{\gamma}}}\qty(\vb{\beta}+\vb{\gamma})!}
					 }
					 {\mat{\Omega}!\vb{\beta}!\vb{\gamma}!}
				\HG[\mat{\Theta}]{\vb{\beta}+\vb{\gamma}}\qty(\va{r})
		~.
	\end{align}
	That is, the sum is over the multi-index matrices $\mat{\Omega}$, with $\vb{\beta}$ being $\vb{\nu}$ minus the row sums of $\mat{\Omega}$, $\vb{\gamma}$ being $\vb{\mu}$ minus the column sums of $\mat{\Omega}$ and such that $\vb{\beta},\vb{\gamma}$ are multi-indices
(consisting of non-negative integers).
\end{lemma}
\begin{proof}
	\begin{align}
		\sum_{\vb{\nu},\vb{\mu}\in\Natural^n} &
			\sqrt{\frac{2^{\abs{\vb{\nu}}+\abs{\vb{\mu}}}}{\vb{\nu}!\vb{\mu}!}}
			\va{x}^{\vb{\nu}}
			\va{y}^{\vb{\mu}}
			\HG[\mat{\Theta}]{\vb{\nu}}\qty(\va{r})
			\HG[\mat{\Theta}]{\vb{\mu}}\qty(\va{r})
        \nonumber
        \\
		&=
			\frac{
					\ee^{-\frac{1}{2}\va{r}^\transpose\mat{\Theta}^{-1}\va{r} + \va{x}^\transpose\mat{\Theta}^{-1}\qty(2\va{r}-\va{x})}
					\ee^{-\frac{1}{2}\va{r}^\transpose\mat{\Theta}^{-1}\va{r} + \va{y}^\transpose\mat{\Theta}^{-1}\qty(2\va{r}-\va{y})}
				 }
				 {\qty(\mpi^n\abs{\mat{\Theta}})^{\frac{1}{2}}}
		\\
		&=
			\frac{\ee^{-\frac{1}{2}\va{r}^\transpose\mat{\Theta}^{-1}\va{r}}}
				 {\qty(\mpi^n\abs{\mat{\Theta}})^{\frac{1}{2}}}
				\ee^{-\frac{1}{2}\va{r}^\transpose\mat{\Theta}^{-1}\va{r} + \qty(\va{x}+\va{y})^\transpose\mat{\Theta}^{-1}\qty[2\va{r}-\qty(\va{x}+\va{y})]}
				\ee^{2\va{x}^\transpose\mat{\Theta}^{-1}\va{y}}
		\\
		&=
			\frac{\ee^{-\frac{1}{2}\va{r}^\transpose\mat{\Theta}^{-1}\va{r}}}
				 {\qty(\mpi^n\abs{\mat{\Theta}})^{\frac{1}{4}}}
			\sum_{\vb{\alpha}\in\Natural^n}
				\sqrt{\frac{2^{\abs{\vb{\alpha}}}}{\vb{\alpha}!}}
				\qty(\va{x}+\va{y})^{\vb{\alpha}}
				\HG[\mat{\Theta}]{\vb{\alpha}}\qty(\va{r})
			\sum_{m\geq 0}
				\frac{\qty(2\va{x}^\transpose\mat{\Theta}^{-1}\va{y})^m}{m!}
		.
	\end{align}
	Denote $\mat{\Theta}^{-1}=\qty[q_{jk}]$, the elements of $\mat{\Theta}^{-1}$, and apply again the multinomial theorem:
	\begin{align}
		\qty(\va{x}+\va{y})^{\vb{\alpha}}
		&=
			\sum_{\substack{
				\vb{\beta}\in\Natural^n,\\
				\mathclap{
					\text{s.t. }
					\vb{\beta}\preceq\vb{\alpha}
				}
			}}
				\binom{\vb{\alpha}}{\vb{\beta}}
				\va{x}^{\vb{\beta}}
				\va{y}^{\vb{\alpha}-\vb{\beta}}
		~,
		\\
		\sum_{m\geq 0}
			\frac{\qty(2\va{x}^\transpose\mat{\Theta}^{-1}\va{y})^m}{m!}
		&=
		\sum_{m\geq 0}
			\frac{\qty[2 \sum_{jk} q_{jk}x_jy_k]^m}{m!}
		=
			\sum_{\substack{
				\mat{\Omega}\in\Natural^{n\times n}
			}}
			\frac{\qty(2\mat{\Theta}^{-1})^{\mat{\Omega}}}{\mat{\Omega}!}
			\qty(\va{x}\va{y}^\transpose)^{\mat{\Omega}}
		~.
	\end{align}
	Equating the powers of $\va{x}$ and $\va{y}$ proves the lemma.
\end{proof}

\cref{lemma_product_of_hg,lemma_offseted} extend well-known results from the univariate case to the multivariate anisotropic case.

\section{Linear Canonical Transform} \label{section_lct}

The linear canonical transform (LCT) generalizes important well-known integral transforms, such as the (fractional) Fourier transform and the Fresnel transform.
The $n$-dimensional LCT (with unitary, angular-frequency kernels) is defined with respect to a matrix $\mat{A}=\qty[\begin{smallmatrix}a\ b\\c\ d\end{smallmatrix}]\in\Complex^{2\times2}$ with
$\abs{\mat{A}}=1$ as
\begin{align}
	\lct[\mat{A}]{f}\qty(\va{\zeta})
		&\triangleq
			\qty(\frac{1}{2\mpi\ii b})^{\frac{n}{2}}
			\ee^{\ii \frac{d}{2b} \va{\zeta}^2}
			\int_{\Real^n} \dd{\va{r}^\prime}
				f\qty(\va{r}^\prime)
				\ee^{-\ii \frac{1}{2 b} \va{r}^\prime\cdot \qty(
					2 \va{\zeta} - a \va{r}^\prime
				)}
		.
	\label{lct}
\end{align}

Our main result in this section follows:
\begin{theorem}[Linear canonical transform of the AHG function] \label{theorem_lct}
Suppose
	$\mat{A}$ is as above. Then
	\begin{align}
		\lct[\mat{A}]{\HG[\mat{\Theta}]{\vb{\nu}}}\qty(\va{\zeta})
			=
				\qty(\frac{1}{\ii b})^{\abs{\vb{\nu}}+\frac{n}{2}}
				\ee^{-\frac{1}{2} \va{\xi}^\transpose \mat{C} \va{\xi}}
				\frac{\abs{\mat{\Xi}}^{\frac{1}{4}}}
					 {\abs{\mat{\Sigma}}^{\frac{1}{2}}\abs{\mat{\Theta}}^{\frac{1}{4}}}
				\HGd[\mat{\Xi}]{\vb{\nu}}\qty(\va{\xi})
		\label{eq_lct_hg_thrm}
		~,
	\end{align}
    where
    \begin{subequations}
    \label{LCT_shorthands}
	\begin{alignat}{3}
		\mat{\Sigma} &= \mat{\Theta}^{-1}-\ii\frac{a}{b}\mat{I}
		~,
		& \qquad\qquad
		\mat{\Xi} &= b^2[2(\mat{\Theta}\mat{\Sigma}\mat{\Theta})^{-1} - \mat{\Theta}^{-1}]
        \label{LCT_shorthands_a}
		~,
		\\
		\mat{C}   &= b^{-1} \mat{\Theta} \qty(b^{-1}\mat{\Sigma}-\ii d \mat{\Sigma}^2) \mat{\Theta} - \mat{\Xi}^{-1}
		~,
		&
		\va{\xi}  &= \mat{\Sigma}^{-1}\mat{\Theta}^{-1}\va{\zeta}
        \label{LCT_shorthands_b}
		~,
	\end{alignat}
    \end{subequations}
	under the conditions that $b\neq 0$ and $\mat{\Sigma},\mat{\Xi}$ both have a positive definite real part.
\begin{remark}
	A sufficient condition for $\Re\mat{\Sigma}\succ 0$ is $a,b\in\Real$ (as $\Re\mat{\Theta}^{-1}\succ 0$).
\end{remark}
\end{theorem}
\begin{proof}
	Take the LCT (with respect to the variable $\va{r}$) of each side of the generating function for $\HG[\mat{\Theta}]{\vb{\nu}}$ (\cref{HG_generating_function}).
	Let $\va{y}=2\mat{\Theta}^{-1}\va{x}-\ii\frac{1}{b}\va{\zeta}$ and rewrite the integral as a multidimensional Gaussian integral with a linear term, which admits a well-known closed-form \cite{Stoof2009}
    (convergence is ensured by $\Re\mat{\Sigma}\succ 0$). Then
	\begin{align}
		\sum_{\vb{\nu}\in\Natural^n}
			\sqrt{\frac{2^{\abs{\vb{\nu}}}}{\vb{\nu}!}}
			\va{x}^{\vb{\nu}}
			&\lct[\mat{A}]{\HG[\mat{\Theta}]{\vb{\nu}}}\qty(\va{\zeta})
        \nonumber 
        \\
		&=
			\qty(\mpi^n\abs{\mat{\Theta}})^{-\frac{1}{4}}
			\lct[\mat{A}]{
				\ee^{-\frac{1}{2}\qty(\va{r}^\prime)^\transpose\mat{\Theta}^{-1}\va{r}^\prime - \va{x}^\transpose\mat{\Theta}^{-1}\va{x} + 2\qty(\va{r}^\prime)^\transpose\mat{\Theta}^{-1}\va{x}}
			}\qty(\va{\zeta})
			\\
		&=
			\qty(\mpi^n\abs{\mat{\Theta}})^{-\frac{1}{4}}
			\qty(\frac{1}{2\mpi\ii b})^{\frac{n}{2}}
			\ee^{\ii \frac{d}{2b} \va{\zeta}^2}
			\ee^{-\va{x}^\transpose\mat{\Theta}^{-1}\va{x}}
			\int_{\Real^n} \dd{\va{r}^\prime}
				\ee^{
					-\frac{1}{2}\qty(\va{r}^\prime)^\transpose
					\mat{\Sigma}
					\va{r}^\prime
					+
					\va{y}^\transpose \va{r}^\prime
				}
			\\
		&= \label{eq_gaussian_integral_thm_lct_proof}
			\qty(\mpi^n\abs{\mat{\Theta}})^{-\frac{1}{4}}
			\frac{1}{\qty(\ii b)^{\frac{n}{2}} \abs{\mat{\Sigma}}^{\frac{1}{2}}}
			\ee^{\ii \frac{d}{2b} \va{\zeta}^2}
			\ee^{-\va{x}^\transpose\mat{\Theta}^{-1}\va{x}}
			\ee^{
				\frac{1}{2}
				\va{y}^\transpose \mat{\Sigma}^{-1} \va{y}
			}
		~.
	\end{align}
	Rewrite the right-hand side above in terms of $\va{x},\va{\xi},\mat{\Xi}$ in the form of the generating function of the dual AHG function (\cref{HG_dual_generating_function}), i.e.:
	\begin{align}
		\sum_{\vb{\nu}\in\Natural^n}
			\sqrt{\frac{2^{\abs{\vb{\nu}}}}{\vb{\nu}!}}
			\va{x}^{\vb{\nu}}
			&\lct[\mat{A}]{\HG[\mat{\Theta}]{\vb{\nu}}}\qty(\va{\zeta})
        \nonumber
        \\
		=&
			\frac{\qty(\mpi^n\abs{\mat{\Theta}})^{-\frac{1}{4}}}
				 {\qty(\ii b)^{\frac{n}{2}} \abs{\mat{\Sigma}}^{\frac{1}{2}}}
			\ee^{-\frac{1}{2} \va{\xi}^\transpose \mat{C} \va{\xi}}
			\ee^{
				-\frac{1}{2} \va{\xi}^\transpose \mat{\Xi}^{-1} \va{\xi} +
				 \qty(-\frac{\ii}{b}\va{x})^\transpose
					\qty[2\va{\xi} - \mat{\Xi}\qty(-\frac{\ii}{b}\va{x})
				 ]}
		\\
		=&
			\frac{1}{\qty(\ii b)^{\frac{n}{2}}}
			\frac{\abs{\mat{\Xi}}^{\frac{1}{4}}}
				 {\abs{\mat{\Sigma}}^{\frac{1}{2}}\abs{\mat{\Theta}}^{\frac{1}{4}}}
			\ee^{-\frac{1}{2} \va{\xi}^\transpose \mat{C} \va{\xi}}
			\sum_{\vb{\nu}\in\Natural^n}
				\sqrt{\frac{2^{\abs{\vb{\nu}}}}{\vb{\nu}!}}
				\qty(-\frac{\ii\va{x}}{b})^{\vb{\nu}}
				\HGd[\mat{\Xi}]{\vb{\nu}}\qty(\va{\xi})
		~.
	\end{align}
	Equating the powers of $\va{x}$ on both sides yields the final result.
\end{proof}

\begin{lemma}[Eigenfunctions of the linear canonical transform] \label{lemma_eigenfunctions_lct}
	If $a=d$, $a^2\neq 1$ and $\sqrt{(a^2-1)b^2}\neq -ab$, then set
    $\alpha=\ii\frac{b^2}{\sqrt{(a^2-1)b^2}+ab}, \beta=\ii\frac{b^2}{\sqrt{(a^2-1)b^2}}$. 
    We have
	\begin{align}
		\lct[\mat{A}]{\HG[\beta\mat{I}]{\vb{\nu}}} &=
			\qty(\frac{1}{\ii b})^{\abs{\vb{\nu}}+\frac{n}{2}}
			\alpha^{\abs{\vb{\nu}}}
			\sqrt{\alpha^n}
			\HG[\beta\mat{I}]{\vb{\nu}}
		~.
	\end{align}
\end{lemma}
\begin{proof}
	$\mat{\Theta}=\beta\mat{I}$, therefore \cref{LCT_shorthands_a,LCT_shorthands_b} become
	\begin{align}
		\mat{\Sigma} = \alpha^{-1}\mat{I} 
		\text{, } &&
		\mat{\Xi} = \alpha^2\mat{\Theta}^{-1} 
		\text{, }&&
		\mat{C}   = 0 
		\text{, } &&
		\mat{\Xi}^{-1}\va{\xi}	= \alpha^{-1}\va{\zeta} 
		.
	\end{align}
	Now consider \cref{eq_lct_hg_thrm} and apply \cref{basic_property_dual,basic_property_analytic_b}:
	\begin{align}
		\frac{\abs{\mat{\Xi}}^{\frac{1}{4}}}
			 {\abs{\mat{\Sigma}}^{\frac{1}{2}}\abs{\mat{\Theta}}^{\frac{1}{4}}}
		\HGd[\mat{\Xi}]{\vb{\nu}}\qty(\va{\xi})
			&=
				\frac{\HG[\mat{\Xi}^{-1}]{\vb{\nu}}\qty(\mat{\Xi}^{-1}\va{\xi})}
					 {\abs{\mat{\Xi}}^{\frac{1}{4}}\abs{\mat{\Sigma}}^{\frac{1}{2}}\abs{\mat{\Theta}}^{\frac{1}{4}}}
			=
				\frac{\HG[\alpha^{-2}\mat{\Theta}]{\vb{\nu}}\qty(\frac{1}{\alpha}\va{\zeta})}
					 {\abs{\mat{\Xi}}^{\frac{1}{4}}\abs{\mat{\Sigma}}^{\frac{1}{2}}\abs{\mat{\Theta}}^{\frac{1}{4}}}
			=
				\alpha^{\abs{\vb{\nu}}}
				\sqrt{\alpha^n}
				\HG[\mat{\Theta}]{\vb{\nu}}\qty(\va{\zeta})
		~,
	\end{align}
	from which \cref{lemma_eigenfunctions_lct} follows.
\end{proof}

The following corollaries follow from \cref{theorem_lct,lemma_eigenfunctions_lct} as well as the basic properties of the AHG functions.

\begin{corollary}[Fourier Transform] \label{cor_ft}
	The LCT reduces to the standard Fourier transform (with unitary, angular frequency kernels) by setting $\mat{A}_\text{FT}=\qty[\begin{smallmatrix}0 & 1 \\ -1 & 0\end{smallmatrix}]$, viz.
		$ \frft*{\HG[\mat{\Theta}]{\vb{\nu}}}
			\triangleq
				\ii^{\frac{n}{2}} \lct*[\mat{A}_\text{FT}]{\HG[\mat{\Theta}]{\vb{\nu}}}
		$.
	In this case the Fourier transform of the anisotropic Hermite-Gauss function is:
	\begin{align}
		\frft{\HG[\mat{\Theta}]{\vb{\nu}}}\qty(\va{\zeta})
			&=
				\qty(-\ii)^{\abs{\vb{\nu}}}
				\HGd[\mat{\Theta}^{-1}]{\vb{\nu}}\!\qty(\va{\zeta})
			=
				\qty(-\ii)^{\abs{\vb{\nu}}}
				\abs{\mat{\Theta}}^{\half}
				\HG[\mat{\Theta}]{\vb{\nu}}\qty(\mat{\Theta}\va{\zeta})
		.
	\end{align}
	In addition,
	\begin{tasks}[label-format=,label=\textrm{\thetheorem.\arabic*},label-width=3em,item-indent=5em](1)
		\task \label[corollary]{fft_eigenfunctions}
			$\HG[\mat{I}]{\vb{\nu}}$ are the eigenfunctions of the Fourier transform with corresponding eigenvalues $(-\ii)^{\abs{\vb{\nu}}}$.
		\task \label[corollary]{fft_order_real_imaginary}
			If $\mat{\Theta}$ is real, then the Fourier transform of an even-order AHG function is purely real, and of an odd-order AHG function purely imaginary.
	\end{tasks}
\end{corollary}

\begin{corollary}[Fractional Fourier Transform] \label{cor_frft}
	Similarly, the LCT also generalizes the fractional Fourier transform (FrFT) of degree $\gamma$ via the
parameter matrix $\mat{A}_\text{FrFT}(\gamma)=\qty[\begin{smallmatrix}\cos\gamma & \sin\gamma \\ -\sin\gamma & \cos\gamma\end{smallmatrix}]$ by
		$ \frft*[\gamma]{\HG[\mat{\Theta}]{\vb{\nu}}}
			\triangleq
				\ee^{\ii \frac{n}{2}\gamma}
				\lct*[\mat{A}_\text{FrFT}(\gamma)]{\HG[\mat{\Theta}]{\vb{\nu}}}
		$.
	The eigenfunctions of the Fractional Fourier transform are $\HG[\mat{I}]{\vb{\nu}}$, with corresponding eigenvalues $e^{-\ii\gamma\abs{\vb{\nu}}}$
\end{corollary}

The fact that the univariate HG functions serve as the eigenfunctions of the (fractional) Fourier transform is well-known.
The corollaries above generalize these results to the multivariate (fractional) Fourier transform and AHG functions.
This has interesting Fourier optics interpretations:
AHG beams remain AHG beams under far-field diffraction (\cref{fft_eigenfunctions}).
Furthermore, only diffracted even-order AHG modes propagate to the far-field while odd-order modes diffract as evanescent waves (consequence of \cref{fft_order_real_imaginary}).

We omit the proof of the following corollary.

\begin{corollary}[Laplace Transform] \label{cor_laplace}
	The (two-sided) Laplace transform is a special case of the LCT with $\mat{A}_\text{L}=\qty[\begin{smallmatrix}0 & \ii \\ \ii & 0\end{smallmatrix}]$, viz.
	\begin{align}
		\lt{\HG[\mat{\Theta}]{\vb{\nu}}}\qty(\va{\zeta})
			\triangleq
				(-2\mpi)^{\frac{n}{2}} \lct[\mat{A}_\text{L}]{\HG[\mat{\Theta}]{\vb{\nu}}}
			&=
				\qty(2\mpi)^{\frac{n}{2}}
				\ii^{\abs{\vb{\nu}}}
				\abs{\mat{\Theta}}^{\half}
				\HG[\mat{\Theta}]{\vb{\nu}}\qty(\ii\mat{\Theta}\va{\zeta})
		.
	\end{align}
	The eigenfunctions of the Laplace transform are
	$
			\HG[\mat{I}]{\vb{\nu}}(\frac{1+\ii}{\sqrt{2}}\va{\zeta})
	$, with corresponding eigenvalues $
		\qty(2\mpi)^{\frac{n}{2}}
		\ii^{\abs{\vb{\nu}}}
		\sqrt{(-\ii)^n}
	$.
\end{corollary}



\section{Wigner-Vile Distribution} \label{section_wvd}

The \emph{Wigner-Vile Distribution} (WVD) is an integral transform that commonly arise in optics and quantum mechanics, useful for processing linear frequency-modulated signals.
The WVD of a $\Real^n\to\Complex$ $L^2$-function $f$ is defined as the following Fourier transform:
\begin{align}
	\wvd{f}\qty(\va{r},\va{\zeta})
		& \triangleq
			\frft{f\qty(\va{r}-\frac{1}{2}\va{\xi}) f^\star\qty(\va{r}+\frac{1}{2}\va{\xi})}\qty(\va{\zeta})
	,
\end{align}
where the FT is taken with respect to the integration variable $\va{\xi}$.

\begin{lemma} \label{wvd_lemma_fft}
	Let $\va{r},\va{\zeta}\in\Real^n$, $\mat{\Theta}\in\Real^{n\times n}$ s.t. $\mat{\Theta}\succ 0$. Then
	\begin{align}
		&\frft{
			\HG[\mat{\Theta}]{\vb{\nu}}\qty(\va{r}-\frac{1}{2}\va{\xi})
			\HG[\mat{\Theta}]{\vb{\mu}}\qty(\va{r}+\frac{1}{2}\va{\xi})
		}\qty(\va{\zeta})
			=
			\qty(4^n\mpi^n \abs{\mat{\Theta}}^3)^{\frac{1}{4}}
			\ee^{-\frac{1}{2}\va{\zeta}^\transpose\mat{\Theta}\va{\zeta}}
			\sum_{\substack{
				\vb{\tau}\preceq\vb{\nu},\\
				\vb{\sigma}\preceq\vb{\mu}
			}}
				\qty(-1)^{\abs{\vb{\nu}-\vb{\tau}}}
        \nonumber
		\\
		&\;\qquad\times
				\ii^{\abs{\vb{\nu}+\vb{\mu}-\vb{\tau}-\vb{\sigma}}}
				\sqrt{
					\binom{\vb{\nu}}{\vb{\tau}}
					\binom{\vb{\mu}}{\vb{\sigma}}
					\binom{\vb{\nu}+\vb{\mu}-\vb{\tau}-\vb{\sigma}}{\vb{\nu}-\vb{\tau}}
				}
				\HG[\mat{\Theta}]{\vb{\tau}}\qty(\va{r})
				\HG[\mat{\Theta}]{\vb{\sigma}}\qty(\va{r})
				\HGd[\mat{\Theta}^{-1}]{\vb{\mu}+\vb{\nu}-\vb{\sigma}-\vb{\tau}}\qty(\va{\zeta})
        \nonumber
		~.
	\end{align}
\end{lemma}
\begin{proof}
	Take the FT of the generating functions:
	\begin{align}
		\sum_{\vb{\nu},\vb{\mu}\in\Natural^n}
			\sqrt{\frac{2^{\abs{\vb{\nu}}+\abs{\vb{\mu}}}}{\vb{\nu}!\vb{\mu}!}}
			&\va{x}^{\vb{\nu}}
			\va{y}^{\vb{\mu}}
			\frft{
				\HG[\mat{\Theta}]{\vb{\nu}}\qty(\va{r}-\frac{1}{2}\va{\xi})
				\HG[\mat{\Theta}]{\vb{\mu}}\qty(\va{r}+\frac{1}{2}\va{\xi})
			}\qty(\va{\zeta})
        \nonumber
		\\
		=&
			\frac{1
			}{\sqrt{\mpi^n\abs{\mat{\Theta}}}}
			\frft
            \Big\{
				\ee^{-\frac{1}{2}\qty(\va{r}-\frac{1}{2}\va{\xi})^\transpose\mat{\Theta}^{-1}\qty(\va{r}-\frac{1}{2}\va{\xi}) + \va{x}^\transpose\mat{\Theta}^{-1}\qty(2\va{r}-\va{\xi}-\va{x})}
        \nonumber \\
        &\qquad\qquad\qquad\qquad\times
				\ee^{-\frac{1}{2}\qty(\va{r}+\frac{1}{2}\va{\xi})^\transpose\mat{\Theta}^{-1}\qty(\va{r}+\frac{1}{2}\va{\xi}) + \va{y}^\transpose\mat{\Theta}^{-1}\qty(2\va{r}+\va{\xi}-\va{y})}
			\Big\}\qty(\va{\zeta})
		\\
		=&
			\frac{1}
				 {\sqrt{\mpi^n\abs{\mat{\Theta}}}}
            \ee^{-\va{r}^\transpose\mat{\Theta}^{-1}\va{r}}
			\ee^{\va{x}^\transpose\mat{\Theta}^{-1}\qty(2\va{r}-\va{x})}
			\ee^{\va{y}^\transpose\mat{\Theta}^{-1}\qty(2\va{r}-\va{y})}
        \nonumber 
        \\
        &\qquad\qquad\qquad\qquad\qquad\times
			\frft{
				\ee^{-\frac{1}{2}\va{\xi}^\transpose\qty(2\mat{\Theta})^{-1}\va{\xi}}
				\ee^{2\va{\xi}^\transpose\qty(2\mat{\Theta})^{-1}\qty(\va{y}-\va{x})}
			}\qty(\va{\zeta})
		\label{eq_lemma_wvd_fft_proof_step1}
		~,
	\end{align}
	complete the square and integrate (in similar fashion to \cref{eq_gaussian_integral_thm_lct_proof}):
	\begin{align}
		\frft{
				\ee^{\va{\xi}^\transpose\qty(2\mat{\Theta})^{-1}\qty[-\frac{1}{2}\va{\xi} + 2\qty(\va{y}-\va{x})]}
			}\qty(\va{\zeta})
		&=
            \qty(\frac{1}{2\mpi})^{\frac{n}{2}}
		    \int \dd{\va{\xi}}
				\ee^{-\frac{1}{2} \va{\xi}^\transpose\qty(2\mat{\Theta})^{-1}\va{\xi} + \va{\xi}^\transpose\va{\zeta}^\prime}
        \\
		&=
				\sqrt{2^n\abs{\mat{\Theta}}}
				\ee^{\qty(\va{\zeta}^\prime)^\transpose\mat{\Theta}\va{\zeta}^\prime}
		\label{eq_lemma_wvd_fft_proof_step2}
	\end{align}
    where we set $\va{\zeta}^\prime=\mat{\Theta}^{-1}\qty(\va{y}-\va{x}) - \ii\va{\zeta}$. 
    The FT always convergences as $\mat{\Theta}^{-1}\succ 0$.
	Then, putting \cref{eq_lemma_wvd_fft_proof_step1,eq_lemma_wvd_fft_proof_step2} together and rewriting the result as the generating functions of
	$\HG[\mat{\Theta}]{\vb{\tau}}(\va{r})$,
	$\HG[\mat{\Theta}]{\vb{\sigma}}(\va{r})$,
	$\HGd[\mat{\Theta}^{-1}]{\vb{\alpha}}(\va{\zeta})$ with variables $\va{x}$, $\va{y}$ and $\va{x}-\va{y}$, respectively, gives
	\begin{align}
		\sum_{\vb{\nu},\vb{\mu}\in\Natural^n}
			\sqrt{\frac{2^{\abs{\vb{\nu}}+\abs{\vb{\mu}}}}{\vb{\nu}!\vb{\mu}!}}
			&\va{x}^{\vb{\nu}}
			\va{y}^{\vb{\mu}}
			\frft{
				\HG[\mat{\Theta}]{\vb{\nu}}\qty(\va{r}-\frac{1}{2}\va{\xi})
				\HG[\mat{\Theta}]{\vb{\mu}}\qty(\va{r}+\frac{1}{2}\va{\xi})
			}\qty(\va{\zeta})
        \nonumber
		\\
		&=
			\qty(\frac{2}{\mpi})^{\frac{n}{2}}
			\ee^{-\va{\zeta}^\transpose\mat{\Theta}\va{\zeta}}
			\ee^{2\ii\qty(\va{x}-\va{y})^\transpose\va{\zeta}}
			\ee^{\qty(\va{x}-\va{y})^\transpose\mat{\Theta}^{-1}\qty(\va{x}-\va{y})}
        \nonumber \\
        &\quad\qquad\qquad\times
			\ee^{-\frac{1}{2}\va{r}^\transpose\mat{\Theta}^{-1}\va{r} +
				 \va{x}^\transpose\mat{\Theta}^{-1}\qty(2\va{r}-\va{x})}
			\ee^{-\frac{1}{2}\va{r}^\transpose\mat{\Theta}^{-1}\va{r} +
				 \va{y}^\transpose\mat{\Theta}^{-1}\qty(2\va{r}-\va{y})}
		\\
		&=
			\qty(4^n\mpi^n \abs{\mat{\Theta}})^{\frac{1}{4}}
			\ee^{-\frac{1}{2}\va{\zeta}^\transpose\mat{\Theta}\va{\zeta}}
			\sum_{\vb{\alpha}\in\Natural^n}
				\sqrt{\frac{2^{\abs{\vb{\alpha}}}}{\vb{\alpha}!}}
				\qty(\va{x}-\va{y})^{\vb{\alpha}}
				\ii^{\abs{\vb{\alpha}}}
				\HGd[\mat{\Theta}^{-1}]{\vb{\alpha}}\qty(\va{\zeta})
        \nonumber \\
        &\qquad\qquad\qquad\qquad\times
			\sum_{\vb{\tau},\vb{\sigma}\in\Natural^n}
				\sqrt{\frac{2^{\abs{\vb{\tau}}+\abs{\vb{\sigma}}}}{\vb{\tau}!\vb{\sigma}!}}
				\va{x}^{\vb{\tau}}
				\va{y}^{\vb{\sigma}}
				\HG[\mat{\Theta}]{\vb{\tau}}\qty(\va{r})
				\HG[\mat{\Theta}]{\vb{\sigma}}\qty(\va{r})
		~.
	\end{align}
	Finally, apply the multinomial theorem, viz.:
	\begin{align}
		\qty(\va{x}-\va{y})^{\vb{\alpha}}
		&=
			\sum_{\substack{
				\vb{\beta}\preceq\vb{\alpha}
			}}
				\binom{\vb{\alpha}}{\vb{\beta}}
				\qty(-1)^{\abs{\vb{\beta}}}
				\va{x}^{\vb{\alpha} - \vb{\beta}}
				\va{y}^{\vb{\beta}}
	\end{align}
	and equate the powers of $\va{x}$ and $\va{y}$, yielding the lemma.
\end{proof}

As any arbitrary $L^2$-function can be expanded in AHG space (\cref{lemma_orthogonality}), by using \cref{wvd_lemma_fft} we can write an expression for the WVD of that function.
In practice, this allows direct computation of the WVD for functions that can be expressed as a superposition of a limited number of AHG functions (e.g., AHG beams).
\begin{theorem}[WVD in AHG space] \label{wvd_theorem}
	Let $\mat{\Theta}\in\Real^{n\times n}$, with $\Re\mat{\Theta}\succ 0$,
	and $f(\va{r})=\sum_{\vb{\nu}} a_{\vb{\nu}} \HG[\mat{\Theta}]{\vb{\nu}}(\va{r})$ be an $\Real^n\to\Complex$ $L^2$-functions expressed via its AHG-basis coefficients, viz. $a_{\vb{\nu}} = \inp*{f}{\HGd[\mat{\Theta}]{\vb{\nu}}}$.
	Then,
	\begin{align}
		&\wvd{f}\qty(\va{r},\va{\zeta})
			=
			\qty(4^n\mpi^n \abs{\mat{\Theta}})^{\frac{1}{4}}
			\ee^{-\frac{1}{2}\va{\zeta}^\transpose\mat{\Theta}\va{\zeta}}
			\sum_{\vb{\nu},\vb{\mu}\in\Natural^n}
				a_{\vb{\nu}}
				a_{\vb{\mu}}^\star
			\sum_{\substack{
				\vb{\tau}\preceq\vb{\nu},\\
				\vb{\sigma}\preceq\vb{\mu}
			}}
				\qty(-1)^{\abs{\vb{\nu}-\vb{\tau}}}
		\nonumber \\
		&\qquad\times
				\ii^{\abs{\vb{\nu}+\vb{\mu}-\vb{\tau}-\vb{\sigma}}}
				\sqrt{
					\binom{\vb{\nu}}{\vb{\tau}}
					\binom{\vb{\mu}}{\vb{\sigma}}
					\binom{\vb{\nu}+\vb{\mu}-\vb{\tau}-\vb{\sigma}}{\vb{\nu}-\vb{\tau}}
				}
				\HG[\mat{\Theta}]{\vb{\tau}}\qty(\va{r})
				\HG[\mat{\Theta}]{\vb{\sigma}}\qty(\va{r})
				\HGd[\mat{\Theta}^{-1}]{\vb{\mu}+\vb{\nu}-\vb{\sigma}-\vb{\tau}}\qty(\va{\zeta})
        \nonumber
		~.
	\end{align}
\end{theorem}
\begin{proof}
	Write
	\begin{align}
		\wvd{f}\qty(\va{r},\va{\zeta})
		&=
			\frft{
				\sum_{\vb{\nu},\vb{\mu}\in\Natural^n}
					a_{\vb{\nu}} a_{\vb{\mu}}^\star
					\HG[\mat{\Theta}]{\vb{\nu}}\qty(\va{r}-\frac{1}{2}\va{\xi})
					\HG[\mat{\Theta}]{\vb{\mu}}\qty(\va{r}+\frac{1}{2}\va{\xi})
			}\qty(\va{\zeta})
	\end{align}
	and apply \cref{wvd_lemma_fft}.
\end{proof}


\bibliographystyle{apalike}
\bibliography{paper}





\end{document}